\documentclass[smallcondensed,final]{svjour3}
\smartqed  

\usepackage{amsmath,amssymb,amsfonts}
\usepackage{theorem}
\usepackage{graphicx}
\usepackage{times,mathptm}
\usepackage{enumerate}
\usepackage{color}
\usepackage{tikz}
\usetikzlibrary{graphs}
\usetikzlibrary{positioning}

\newcommand{\implik}{\Longrightarrow}

\newcommand{\ekviv}{\Longleftrightarrow}

\newcommand{\mnimpl}{\mbox{$
    \mathop{-\raise.8pt\hbox{\mathsurround=0pt$\!\scriptstyle\circ$}}$}}

\usepackage{bm}

\newtheorem{observation}[theorem]{Observation}

\begin{document}
\title{Dynamic logic assigned to automata}
\author{Ivan~Chajda\thanks{{Both authors acknowledge the support by a bilateral project 
New Perspectives on Residuated Posets  financed by  
Austrian Science Fund (FWF): project I 1923-N25, 
and the Czech Science Foundation (GA\v CR): project 15-34697L}. } \and Jan~Paseka}

\institute{I.~Chajda \at Department of Algebra and Geometry,
                               Faculty of Science, 
                               Palack\'y University Olomouc, 17. listopadu 12,
                               Olomouc, 771 46, Czech Republic\\
					\email{ivan.chajda@upol.cz}
\and J.~Paseka \at Department of Mathematics and Statistics,
Faculty of Science, Masaryk University, 
{Kotl\'a\v r{}sk\' a\ 2}, 611~37 Brno, Czech Republic\\
\email{paseka@math.muni.cz}}

\maketitle

\begin{abstract}
A dynamic logic ${\mathbf B}$ can be assigned to every 
automaton ${\mathcal A}$ without regard if ${\mathcal A}$ is deterministic or nondeterministic. 
This logic enables us to formulate observations on ${\mathcal A}$ in the form of composed 
propositions and, due to a transition functor $T$, it captures the dynamic behaviour of ${\mathcal A}$. 
There are formulated conditions under which the automaton ${\mathcal A}$ can be recovered by means of ${\mathbf B}$ and $T$.\end{abstract}

\keywords{dynamic logic \and automaton \and state-transition relation 
\and transition functor \and modal functor}

\subclass{03B60 \and 03D05\and  68S05}
\section{Introduction}

The aim of the paper is to assign a certain logic to a given automaton 
without regard to whether it is deterministic or nondeterministic. This logic 
has to be dynamic in the sense to capture dynamicity of working 
automaton. We consider an {\em automaton} as 
${\mathcal A}=(X,S,R)$, where $X$ is a non-empty set of {\em inputs}, 
$S$ is a non-empty set of {\em states} and 
$R\subseteq X\times S\times S$ is the set of {\em labelled  transitions}. 
In this case we say that $R$ is a {\em state-transition relation} 
and it is considered as a dynamics of ${\mathcal A}$. Hence,  
the automaton ${\mathcal A}$ can be visualized as a graph 
whose vertices are states and edges denote (possibly multiple) transitions 
$s\xrightarrow{x} t$ 
from one state $s$ to another state $t$ provided an input $x$ is coming; 
this is visualized by a label $x$ on the edge $(s,t)$. 
In particular, 
motivated by the above considerations and e.g. by the paper \cite{perinotti} 
where a denumerable set of vertices is used in studying quantum automata 
to recover  the Weyl, Dirac and Maxwell dynamics in the relativistic limit we 
have to assume that the sets $X$ and $S$ can have arbitrarily large cardinality.

Any physical system can be in some sense  considered as an automaton. 
Its states are then states of the automaton and the transitition relation is 
a transition of a physical system from a given state to an admissible one. 
It should be noted that a quantum physical system is nondeterministic since  
particles can pass through a so-called superposition, i.e., they may 
randomly select a state from the set of admissible states. 



On the other hand, we often formulate certain propositions on an 
automaton ${\mathcal A}$ and deduce conclusions about 
the behaviour of ${\mathcal A}$ in the present (i.e., a {\em description}) or 
in a (near) future (i.e., a {\em forecast}). It is apparent that for this aim 
we need a certain logic which is derived from a given automaton 
and which enables us to formulate propositions on ${\mathcal A}$  
and to deduce conclusions and consequences. Due to the 
mentioned dynamics of  ${\mathcal A}$, our logic ${\mathbf B}$ 
should contain a tool for a certain dynamics. 
This tool will be called a {\em transition functor}. This 
transition functor will assign to every proposition $p\in {\mathbf B}$ 
and input $x\in X$ another proposition $q$. 
In a certain case, this functor can be considered as a modal functor 
with one more input from $X$. The above mentioned approach 
has a sense if our logic ${\mathbf B}$ with a transition functor $T$ 
enables us to reconstruct the dynamics of a given automaton ${\mathcal A}$.  
One can compare our approach with the approach from  \cite{perinotti} 
where an automaton can be represented by an operator over a Hilbert space 
or  with the approach from \cite{yongming} or \cite{mendivil} where the role of the transition functor 
is played by a map from $S$ to $({\mathbf M}^{S})^{X}$ where 
${\mathbf M}$ is a bounded lattice of  truth-values 
or by a map from $S$ to $({[0,1]}^{S})^{X}$.

In what follows, we are going to involve  a systematic approach how 
to reach such a transition functor and the logic ${\mathbf B}$ such that 
the reconstruction of the state-transition relation $R$ is possible. 
Since the conditions of our approach are formulated 
in a pure algebraic way, we need to develop an algebraic background 
(see e.g. also in \cite{Blyth}).
It is worth noticing that the  transition functor will be 
constructed formally in a similar way as tense operators introduced 
by J.~Burgess \cite{burges} for the classical logic 
and developped by the authors for several non-classical logics, see \cite{dyn}, \cite{dem} and \cite{doa}, 
and also the monograph \cite{monochapa}.
Because we are not interested in outputs of the  automaton ${\mathcal A}$, 
we will consider  ${\mathcal A}$ as the so-called {\em acceptor} only.

It is worth noticing that certain (temporal) logics assigned to automata were already 
investigated by several authors, see e.g. the seminal papers on temporal 
logics for programs by Vardi  \cite{vardibuchi}, \cite{vardilinear}, the papers \cite{dixon,konur}
and the monograph \cite{fisher}  for additional results and references. 
However, our approach is different. Namely, our logic assigned to an automaton 
is equipped with the so-called transition operator which makes the logic 
to be dynamic.

Besides of the previous, the observer or a user  of an automaton can formulate 
propositions revealing our knowledge about  it depending on the input. The truth-values 
of these propositions depend on states and inputs 
and let us assume that these propositions can acquire only two values, 
namely either TRUE of FALSE.   For example, 
if we fix an input $x\in X$, the proposition $p/x$ can be  true if the 
automaton ${\mathcal A}$ is in the state 
$s$ but false if ${\mathcal A}$ is not in the state $s$. Hence, for each 
state $s\in S$ we can evaluate the truth-value of $p/x$, it is denoted 
by $p/x(s)$. As mentioned above, $p/x(s)\in \{0, 1\}$ where 
$0$ indicates the truth-value FALSE and $1$ indicates TRUE. 

Denote by $B$ the set of propositions about the automaton ${\mathcal A}$ 
formulated by the observer. 
We can introduce the order $\leq$ on $B$ as follows: 
$$
\text{for}\ p,q\in B, p\leq q\ \text{if and only if}\ 
p(s)\leq q(s)\ \text{for all}\ s\in S.
$$
One can immediately check that the contradiction, i.e., the proposition 
with constant truth-value $0$, is the least element and the tautology, i.e., 
the proposition with the constant truth-value $1$ is the greatest element of the 
partially ordered set $(B;\leq)$; this fact will be expressed by the notation  
${\mathbf B}=(B;\leq, 0, 1)$ for the bounded partially ordered set of propositions about 
the automaton ${\mathcal A}$. 

We summarize our description as follows:
\begin{enumerate}[{-}]
\item every automaton ${\mathcal A}$ will be identified with the 
triple $(B,X, S)$, where $B$ is the set of propositions about ${\mathcal A}$, $X$ is the set 
of possible inputs  and $S$ is the set of states on ${\mathcal A}$; 
\item we are given   a  set of labelled transitions $R\subseteq X\times S\times S$  such that, 
for an input $x\in X$, ${\mathcal A}$ can go from $s$ to $t$ provided 
$(x, s,t)\in R$; 
\item the set $B$ is partially ordered by values of propositions as shown above.
\end{enumerate}

If $s\xrightarrow{x} t_1$ and $s\xrightarrow{x} t_2$ yields $t_1=t_2$ for all $s, t_1, t_2\in S$ and $x\in X$ 
we say that ${\mathcal A}$ is a {\em deterministic automaton}. 
If ${\mathcal A}$ is not deterministic we say that it is {\em nondeterministic}.

To shed light on the previous concepts, 
let us present the following example.

\begin{example}\label{firef}\upshape  \label{expend1} 
At first, let us present a very simple automaton ${\mathcal A}$ 
describing a SkyLine Terminal Transfer Service at an airport between Terminals 1 and 2. 
The SkyLine train  is housed, repaired and maintained in  the engine shed and the only way how to get there 
is through Terminal 2.

The observer can distinguish three states as follows:
\begin{enumerate}[{-}]
\item $s_1$ means that the SkyLine train is in  Terminal 1, 
\item $s_2$ means that the  SkyLine train is in  Terminal 2, 
\item $s_3$ means that the SkyLine train is in the engine shed.
\end{enumerate} 

There are two possible actions:

\begin{enumerate}[{-}]
\item $x_1$ means that the passengers entered the   SkyLine train, 
\item $x_2$ means that the  SkyLine train has to be moved to the engine shed.
\end{enumerate} 

If the  SkyLine train is in Terminal 1 or in Terminal 2 then, after the passengers entered it, it moves to the other terminal. 
If the  SkyLine train is  in Terminal 2 then, after the request that the  SkyLine train has to be moved to the engine shed is issued, it moves 
to the engine shed. If the  SkyLine train is  in  the engine shed then, regardless of what action is requested, it stays there.

The set $R$ of labelled transitions on the set  $S=\{s_1, s_2, s_3\}$ 
of states under actions from the set $X=\{x_1, x_2\}$ is of the form 
$$
R=\{(x_1,s_1, s_2), (x_1,s_2,  s_1), (x_1,s_3,s_3), (x_2,s_2, s_3),   
(x_2,s_3, s_3)\}
$$
and it can be vizualized as follows.


\tikzset{
my loop/.style={to path={
.. controls +(80:1) and +(100:1) .. (\tikztotarget) \tikztonodes}},
my state/.style={circle,draw}}

 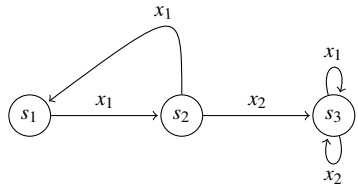
\begin{figure}[h]
\centering
\begin{tikzpicture}[shorten >=1pt,node distance=2cm,auto]
\node[my state] (s_1)  at (0,0) {$s_1$};
\node[my state] (s_2) [right of=s_1] {$s_2$};
\node[my state] (s_3) [right of=s_2] {$s_3$};
\draw [->] (s_1) -- (s_2) node[pos=.5,sloped,above] {$x_1$};
\draw [->] (s_2) -- (s_3) node[pos=.5,sloped,above] {$x_2$};
\draw[->] (s_2) .. controls +(up:1.5cm)  ..   (s_1) node[pos=.5,sloped,above]  {$x_1$};
\path (s_3) edge [loop above]  node {$x_1$}  (s_3); 
\path (s_3) edge [loop below]  node {$x_2$} (s_3);
\end{tikzpicture}
\caption{The transition graph of $R$}\label{Fig2trx}
\end{figure}

The set $B=\{0, p, q, r, p', q', r', 1\}$ of possible
propositions $B$ about the automaton ${\mathcal A}$ is as follows:

\begin{enumerate}[{-}]
\item $0$ means that the SkyLine train    is in no state of $S$, 
\item $p$ means that the SkyLine train   is in Terminal 1, 
\item $q$ means that the SkyLine train   is in Terminal 2, 
\item $r$ means that the  SkyLine train is in the engine shed, 
\item $1$ means that the  SkyLine train  is in at least one state of $S$. 
\end{enumerate}

Considering ${\mathbf B}$ as a classical logic 
(represented by a Boolean algebra $(B; \vee, \wedge, ', 0, 1)$), 
we can apply logical connectives conjunction $\wedge$, disjunction $\vee$, 
negation $'$ and implication $\implik$ to create new propositions about 
${\mathcal A}$. In our case, we can get e.g. $p'=q\vee r$ which 
means that the SkyLine train   is  either  in Terminal 2 
or  in the  engine shed, 
etc. Altogether, we obtain eight propositions. 
We may identify $\mathbf B$ with the Boolean algebra $\{0, 1\}^S$ as follows: 
\medskip
\begin{center}
\begin{tabular}{l l l l}
$0=(0,0,0)$,& $p=(1, 0, 0)$,&$q=(0, 1, 0)$, &$r=(0, 0, 1)$, \\
$p'=(0,1, 1)$,& $q'=(1, 0, 1)$,&
$r'=(1, 1, 0)$,& $1=(1,1,1)$.
\end{tabular}
\end{center}

The interpretation of propositions from $B$ is as follows: 
for any $\alpha\in B$, 
 $\alpha$ is true in the state $s_i$ of the automaton ${\mathcal A}$
if and only if  $\alpha(s_i)=1$.
\end{example}

\section{Algebraic tools}

For the above mentioned construction of 
a suitable logic with a transition functor and the 
reconverse of the given relation, we recall the 
following necessary algebraic tools and results in this section.

Let $S$ be a non-empty set. Every subset $R\subseteq S\times S$ is called 
a {\em relation on $S$} and we say 
that the couple $(S, R)$ is a {\em transition frame}.
The fact that $(s, t)\in R$ for $s, t\in S$  is expressed by the notation $s \mathrel{R} t$. 

Let $A$ be a non-empty set. A relation on $A$ is called a {\em partial order} if it 
is  reflexive, antisymmetric and transitive.  In what follows, partial order  will be 
denoted by the symbol $\leq$ and the 
 pair $\mathbf A=(A;\leq)$ will be referred to as a {\em partially ordered set} (shortly a {\em poset}).

Let $(A;\leq)$ and $(B;\leq)$ be partially ordered sets, $f, g\colon A\to B$ mappings. 
We write $f\leq g$ if $f(a)\leq g(a)$, for all $a\in A$.
A mapping $f$ is called {\em order-preserving} or 
 {\em monotone} if $a, b \in A$ and $a \leq b$ together 
imply $f(a) \leq f(b)$ 
and {\em order-reflecting}  if 
$a, b \in A$ and $f(a) \leq f(b)$ together imply $a \leq b$. 
A bijective order-preserving and order-reflecting mapping 
$f\colon A\to B$  is called an 
{\em isomorphism} and then we say that the partially ordered sets 
$(A;\leq)$ and $(B;\leq)$ are {\em isomorphic}.

Let $(A;\leq)$ and $(B;\leq)$ be partially ordered sets. A mapping $f\colon A\to B$ is 
called {\em residuated} if there exists a mapping $g\colon B\to A$ 
such that $f(a)\leq b\ \text{if and only if}\ a\leq g(b)$ 
for all $a\in A$ and $b\in B$. In this situation, we say that 
$f$  and $g$ form a {\em residuated pair} or that 
the pair $(f,g)$ is a (monotone) {\em Galois connection}. The role of Galois connections is essential for our constructions.

If a partially ordered set $\mathbf A$ has 
both a bottom and a top element, it will be called {\em bounded}; the appropriate 
notation for a bounded partially ordered set is $(A;\leq,0,1)$.  
Let $(A;\leq,0,1)$ and $(B;\leq,0,1)$ be bounded partially ordered sets. A {\em morphism} 
$f\colon A\to B$ {\em of bounded  partially ordered sets} 
is an order, 
top element and bottom element preserving map.

We can take the following useful result from \cite[Observation 1]{dyn}.

\begin{observation}[\cite{dyn}]\label{obsik} Let\/ $\mathbf A$ and\/ $\mathbf M$  be 
bounded partially ordered sets, $S$ a non-empty set, and 
$h_{s}\colon A\to M, s\in S$, morphisms of bounded partially ordered sets.
The following conditions are equivalent:
\begin{enumerate}
\item[{\rm(i)}] \(((\forall s \in S)\, h_{s}(a)\leq h_{s}(b))\implies a\leq
b\) for any elements \(a,b\in A\);
\item[{\rm(ii)}] The map $i_{{}{\mathbf A}}^{S}\colon A \to M^{S}$ defined by 
$i_{{}{\mathbf A}}^{S}(a)=(h_s(a))_{s\in S}$ for all $a\in A$ is order reflecting.
\end{enumerate}
\end{observation}
We then say that $\{h_{s}\colon A\to M; s\in S\}$ is a 
{\em full set of order-preserving maps with respect to} $M$. 
Note that we may in this case 
identify $\mathbf A$ with a bounded subposet   of  $\mathbf{M}^S$ since 
$i_{{}{\mathbf A}}^{S}$ is an order reflecting 
morphism alias {\em embedding} of bounded partially ordered sets. For any $s\in S$ and any 
$p=(p_t)_{t\in S}\in {M}^S$ we denote by $p(s)$ the $s$-th projection $p_s$. 
Note that $i_{{}{\mathbf A}}^{S}(a)(s)=h_s(a)$ for all $a\in A$ and all $s\in S$.

\section{Transition frames and transition operators}

The aim of this section is to recall a construction of two operators 
on partially ordered sets derived by means of a given relation 
and a construction of relations induced by these operators. 
For more details see the paper \cite{transop}.

In what follows, let  $\mathbf{M}=(M;\leq,0, 1)$ be a bounded partially ordered set  and
the bounded subposets  ${\mathbf{A}}=(A;\leq,0, 1)$ and 
${\mathbf{B}}=(B;\leq,0, 1)$  of  $\mathbf{M}^S$  will play the role 
of possibly different logics of propositions pertaining 
to our automaton ${\mathcal A}$, a corresponding set of states $S$,  and 
a state-transition relation $R$ on $S$. 
The operator $T_R\colon B\to {M}^S$ will prescribe to a proposition $b\in B$ 
about ${\mathcal A}$ a new proposition $T_R(b)\in {M}^S$ such that the 
truth value of $T_R(b)$ in state $s\in S$ is the greatest truth value that is smaller or equal 
than the corresponding truth values of $b$ in all states that can be reached from $s$. 
If there is no such state  the 
truth value of $T_R(b)$   in state $s$ will be $1$. Similarly, 
the operator $P_R\colon A\to {M}^S$ will prescribe to a proposition $a\in A$ 
about ${\mathcal A}$ a new proposition $P_R(a)\in {M}^S$ such that the 
truth value of $P_R(a)$ in state $t\in S$ is the smallest truth value that is greater or equal 
than the corresponding truth values of $b$ in all states such that $t$ can be reached from them. 
If there is no such state  the 
truth value of $P_R(a)$   in state $t$ will be $0$.

Specifically, if $M=\{ 0,1\}$  then $T_R(b)$ is true in state 
$s$ if and only if there is  no state $t\in S$  that can be reached from $s$ 
and $b$ is false in $t$, and $P_R(a)$ is false  in state $t$ 
if and only if there is  no state $s\in S$  such that $t$ can be reached from $s$ 
and $b$ is true in $s$.

Consider  a complete lattice  $\mathbf M=(M;\leq,{}0, 1)$ and  
let  $\mathbf{A}=({A};\leq, 0,1)$ and 
$\mathbf{B}=({B};\leq,$ $0,1)$  be  bounded partially ordered sets 
 with a full set $S$  of morphisms  of bounded partially ordered sets 
  into a  non-trivial  complete lattice   $\mathbf{M}$.  
We may assume that 
$\mathbf{A}$ 
and   $\mathbf{B}$ are bounded subposets of\/ $\mathbf{M}^{S}$. 
Further, let  $(S,R)$ be a transition  frame.

Define 
mappings $P_R:A\to {M}^S$ and $T_R:B\to {M}^S$
  as follows: For all $b\in B$ and all $s\in S$,  
  
\begin{equation}\begin{array}{c}\mbox{$T_R(b)(s)=\bigwedge_{M}\{b(t)\mid s R t\} $}\phantom{.} \tag{$\star$}
 \end{array}
\label{eqn:RTD}
\end{equation}
\noindent{}and, for all $a\in A$ and all $t\in S$,  

\begin{equation}
\begin{array}{c}
\mbox{${P}_R(a)(t)=\bigvee_{M}\{a(s)\mid s R  t\} $}{.} \tag{$\star\star$}
 \end{array}
\label{eqn:RPD}
\end{equation}

Then we say that ${T}_R$ ($P_R$) is an {\em upper transition functor} 
({\em lower transition functor})   {\em constructed by means of the  transition frame}  $(S,R)$, 
respectively.  We have that ${T}_R$ is an order-preserving map such that $T_R(1)=1$ and similarly, 
 ${P}_R$ is an order-preserving map such that $P_R(0)=0$.


As an illustration of our approach we present the following example.

\begin{example} \label{expend2} Consider the 
automaton ${\mathcal A}$ and the set of propositions $B$ 
of Example \ref{firef}. 
Then $R=\{x_1\}\times R_{x_1}\cup \{x_2\}\times R_{x_2}$ where 
$R_{x_1}=\{(s_1, s_2), (s_2,  s_1), (s_3,s_3)\}\ \text{and}\ 
R_{x_2}=\{(s_2, s_3),   (s_3, s_3)\}. $

Using our formulas $(\star)$ and $(\star\star)$, we can 
compute the  upper transition functors 
$T_{R_{x_1}}$, $T_{R_{x_2}}\colon B\to 2^{S}$ 
and the lower transition functors 
$P_{R_{x_1}}$, $P_{R_{x_2}}\colon B\to 2^{S}$ as follows: 

\medskip
\begin{center}
\begin{tabular}{c c}
\begin{tabular}{l l}
$T_{R_{x_1}}(0)=0$,&$T_{R_{x_1}}(1)=1$,\\
 $T_{R_{x_1}}(p)=q$,& $T_{R_{x_1}}(p')=q'$,\\
$T_{R_{x_1}}(q)=p$,& $T_{R_{x_1}}(q')=p'$,\\
$T_{R_{x_1}}(r)=r$,&$T_{R_{x_1}}(r')=r'$,\\
\end{tabular}&
\begin{tabular}{l l}
$T_{R_{x_2}}(0)=p$,&$T_{R_{x_2}}(1)=1$,\\
 $T_{R_{x_2}}(p)=p$,& $T_{R_{x_2}}(p')=1$,\\
$T_{R_{x_2}}(q)=p$,& $T_{R_{x_2}}(q')=1$,\\
$T_{R_{x_2}}(r)=1$,&$T_{R_{x_2}}(r')=p$,\\
\end{tabular}\\
\phantom{\large R}&\\
\begin{tabular}{l l}
$P_{R_{x_1}}(0)=0$,&$P_{R_{x_1}}(1)=1$,\\
 $P_{R_{x_1}}(p)=q$,& $P_{R_{x_1}}(p')=q'$,\\
$P_{R_{x_1}}(q)=p$,& $P_{R_{x_1}}(q')=p'$,\\
$P_{R_{x_1}}(r)=r$,&$P_{R_{x_1}}(r')=r'$,\\
\end{tabular}&
\begin{tabular}{l l}
$P_{R_{x_2}}(0)=0$,&$P_{R_{x_2}}(1)=r$,\\
 $P_{R_{x_2}}(p)=0$,& $P_{R_{x_2}}(p')=r$,\\
$P_{R_{x_2}}(q)=r$,& $P_{R_{x_2}}(q')=r$,\\
$P_{R_{x_2}}(r)=r$,&$P_{R_{x_2}}(r')=r$.\\
\end{tabular}
\end{tabular}
\end{center}
\medskip 


E.g., $T_{R_{x_1}}(q)=p$ means that if the Skyline train is in Terminal 1 
 then, after any possible transition under the action that 
the passengers entered  the Skyline train,  it will change to  Terminal 2, 
and  $T_{R_{x_1}}(q')=p'$ means that if the Skyline train is in Terminal 2  or 
in the  engine shed then, after any possible transition under the action that 
the passengers entered  the Skyline train,  it will be in Terminal 1 
or  in the engine shed. Similarly, $T_{R_{x_2}}(1)=1$ 
means that if  the Skyline train  is in at least one state of $S$ then, 
after any possible transition under the action that the  SkyLine train 
has to be moved to the engine shed,  it will be 
 in at least one state of $S$, and $T_{R_{x_2}}(p)=p$ means that  
 if the Skyline train is  in Terminal 1
  then, after any possible transition under the action that the  SkyLine train 
has to be moved to the engine shed (which can be done only at Terminal 2 or 
at the engine shed),  it will stay in Terminal 1. 
\end{example}

Let $P:A\to B$ and $T:B\to A$ be morphisms of  partially ordered sets, $(A;\leq)$ and $(B;\leq)$ 
subposets of \/ $\mathbf{M}^{S}$.
Let us define the relations 
\begin{equation}R_T=\{(s, t)\in S\times S\mid (\forall b\in B) (T(b)(s)\leq b(t))\} \tag{$\dagger$}
\label{eqn:RT}
\end{equation}
and 
\begin{equation}
R^{P}=\{(s, t)\in S\times S\mid (\forall a\in A) (a(s)\leq P(a)(t))\}.\tag{$\dagger\dagger$}
\label{eqn:RP}
\end{equation}

\bigskip

The relations $R_T$ and  $R^{P}$ 
on $S$  
will be called the 
{\em upper $T$-induced relation by ${\mathbf M}$} 
(shortly {\em $T$-induced relation by ${\mathbf M}$}) and 
{\em lower $P$-induced relation  by ${\mathbf M}$} 
(shortly {\em $P$-induced relation by ${\mathbf M}$}),
respectively.

\begin{example} \label{expend3} Consider the automaton ${\mathcal A}$ 
of Example \ref{expend1}.
Let  $P$ be a restriction of the operator $P_{R_{x_2}}$ 
of Example \ref{expend2} and  let 
$T$ be a restriction of the operator $T_{R_{x_2}}$ of the same example. 
Let us compute $R_T$ and 
$R^{P}$. We have $R_T=R^{P}=\{(s_2, s_3), (s_3, s_3)\}$. Hence 
the transition relation $R_{x_2}$ of Example \ref{expend2} coincides with our induced transitions 
relations $R_T$ and $R^{P}$. We can see from above that the operator $T_{R_{x_2}}$ bears the 
maximal amount of information about the transition relation $R_{x_2}$ 
on the subposet of all fixpoints of  $P_{R_{x_2}}\circ T_{R_{x_2}}$. 
The same conclusion holds for the operator $P_{R_{x_2}}$. 
\end{example}

Now, let let $(S, R)$ be a transition frame and $T_R$, $P_R$  functors 
constructed by means of the  transition frame  $(S,R)$. We can ask under what conditions 
the relation $R$ coincides with the relation $R_{T_R}$ constructed as in 
($\dagger$)  or with the relation $R^{P_R}$ constructed as in 
($\dagger\dagger$). 
If this is the case we say that $R$ {\em is recoverable from} $T_R$ or  that 
$R$ {\em is recoverable from} $P_R$. We say that 
$R$ is {\em recoverable} if it is recoverable both from $T_R$  
and  $P_R$.

\begin{example} \label{expend4} Consider the automaton ${\mathcal A}$ 
of Example \ref{expend1}.
Let us put $A=B=\{0, 1\}^S$. Let $P\colon \{0, 1\}^S\to \{0, 1\}^S$ 
and $T\colon \{0, 1\}^S\to \{0, 1\}^S$ be morphisms of  partially ordered sets given as follows:
\medskip
\begin{center}
\begin{tabular}{@{}l @{\,}l@{\,} l@{\,} l@{\,} l@{\,} l@{\,} l@{\,} l}
$T(0)=0$,& $T(p)=q$,&$T(q)=p$, &$T(r)=r$, &$T(p')=q'$,& $T(q')=p'$,&$T(r')=r'$,&$T(1)=1$,\\
$P(0)=0$,& $P(p)=q$,&$P(q)=p$, &$P(r)=r$,& $P(p')=q'$,& $P(q')=p'$,&$P(r')=r'$,& $P(1)=1$.
\end{tabular}
\end{center}
Note that $P$ coincides with the operator $P_{R_{x_1}}$ of Example \ref{expend2}, and 
$T$ coincides with the operator $T_{R_{x_1}}$ of the same example.
We have $R_T=R^{P}=\{(s_1, s_2), (s_2, s_1),(s_3, s_3)\}$. 
The transition relation $R_{x_1}$ of Example \ref{expend1} coincides with our induced transitions 
relations $R_T$ and $R^{P}$. 
\end{example}

The connection between relations induced by means of transition functors 
$T$ and $P$ is shown in the following lemma and theorem. 

 \begin{lemma}\cite{transop}\label{xreldreprest}
Let  $\mathbf{M}$ be a  non-trivial  complete lattice and  $S$ a non-empty set 
such that  $\mathbf{A}$ and $\mathbf{B}$ are bounded subposets of\/ $\mathbf{M}^{S}$. 
 Let $P:A\to {M}^{S}$ and $T:B\to {M}^{S}$ 
be morphisms of partially ordered sets such that, for all $a\in A$ and all $b\in B$, 
$$
P(a)\leq b\ \ekviv\ a\leq T(b).
$$
\begin{enumerate}[{\rm(a)}]
\item If $P(A)\subseteq B$ then $R_T\subseteq R^{P}$.
\item If $T(B)\subseteq A$ then $R^{P}\subseteq R_T$.
\item If $P(A)\subseteq B$ and $T(B)\subseteq A$ then $R_T= R^{P}$.
\end{enumerate}
\end{lemma}

Among other things, the following theorem shows that if 
a given transition relation $R$ can be recovered by the upper transition functor 
then, under natural conditions, it can be recovered by the lower 
transition functor and vice versa.
 
 \begin{theorem}\cite{transop}\label{reldreprest}
Let  $\mathbf{M}$ be a  non-trivial  complete lattice and 
$(S,R)$ a transition frame. Let $\mathbf{A}$ and $\mathbf{B}$ be  
bounded subposets of\/ $\mathbf{M}^{S}$.  Let $P_R:A\to {M}^{S}$ and $T_R:B\to {M}^{S}$ 
be functors  {constructed by means of the transition frame} $(S,R)$.
Then, for all $a\in A$ and all $b\in B$, 
$$
P_R(a)\leq b\ \ekviv\ a\leq T_R(b).
$$
Moreover,  the following holds.
\begin{enumerate}[\rm(a)] 
\item Let for all $t\in S$ exist an element $b^t\in B$ such that, for all $s\in S$, $(s,t)\notin R$, we have 
$\bigwedge_{M}\{u(b^{t})\mid  s R u\}\not\leq t(b^{t})\not =1$. Then 
$R=R_{T_R}$. 
\item Let for all $s\in S$ exist an element $a^s\in A$ such that, for all $t\in S$, $(s,t)\notin R$, we have 
$\bigvee_{M}\{u(a_{s})\mid  u R t\}\not\geq s(a^{s})\not =0$. Then 
 $R=R^{P_R}$. 
 \item If $R=R_{T_R}$ and $T_R(B)\subseteq A$ then $R=R_{T_R}=R^{P_R}$. 
 \item If $R=R^{P_R}$ and $P_R(A)\subseteq B$ then $R=R_{T_R}=R^{P_R}$.
\end{enumerate}
\end{theorem}

The following corollary of Theorem \ref{reldreprest} shows that if the 
set $B$ of propositions on the system $(B,S)$ is large enough, i.e., if it 
contains the full set  $\{0,1\}^S$  then the transition relation 
$R$ can be recovered by each of the transition functors.

 \begin{corollary}\cite{transop}\label{fcorreldreprest}
Let  $\mathbf{M}$ be a  non-trivial  complete lattice and 
$(S,R)$ a transition frame. Let 
$\mathbf{B}$ be a bounded subposet of\/ $\mathbf{M}^{S}$ 
such that $\{0,1\}^{S}\subseteq B$.  Let $P_R:B\to {M}^{S}$ and $T_R:B\to {M}^{S}$ be 
functors  {constructed by means of the transition frame} $(S,R)$.
Then $R=R^{P_R}=R_{T_R}$. 
\end{corollary}

\section{The labelled transition functor characterizing the automaton}

The aim of this section is to derive the logic $\mathbf B$ 
with transition functors corresponding to a given automaton 
${\mathcal A}=(X,S,R)$. This logic $\mathbf B$  will be 
represented via the partially ordered set of its propositions. 
In the rest of the paper, truth-values of our logic $\mathbf B$ 
will be considered to be from the complete lattice $\mathbf M$.
Thus $\mathbf B$ will be a bounded subposet 
of ${\mathbf M}^S$ for the complete lattice ${\mathbf M}$ 
of  truth-values.

Let us consider an automaton  ${\mathcal A}=(X,S,R)$. Clearly, $R$ can be written in the following form 
$$
R=\bigcup_{x\in X}\{x\}\times R_{x}
$$
where $R_x\subseteq S\times S$ for all $x\in X$. Hence, 
for all $x\in X$, using our formulas $(\star)$ and $(\star\star)$, 
we obtain the  upper transition functor $T_{R_x}\colon B\to M^{S}$ 
and the lower transition functor $P_{R_x}\colon B\to M^{S}$. 
It follows that we have functors $T_R=(T_{R_{x}})_{x\in X}\colon B\to (M^{S})^{X}$  
and $P_R=(P_{R_{x}})_{x\in X}\colon B\to (M^{S})^{X}$. We say 
that  $T_R$ is the {\em labelled upper transition functor constructed by means of ${\mathcal A}$} 
and $P_R$ is the {\em labelled lower transition functor constructed by means of ${\mathcal A}$}. 
Note that any mapping $T\colon B\to (M^{S})^{X}$ corresponds 
uniquely to a mapping $\widetilde{T}\colon X\times B\to M^{S}$ such that, 
for all $x\in X$, $T=(\widetilde{T}(x,-))_{x\in X}$. Hence, $T_R$ and $P_R$ 
will play the role of our transition functor.

Now, let $P=(P_x)_{x\in X}:B\to ({M}^{S})^{X}$ and 
$T=(T_x)_{x\in X}:B\to ({M}^{S})^{X}$ be morphisms of partially ordered sets. 
For all $x\in X$, 
let $R^{P_x}$ be the lower $P_x$-induced relation by $\mathbf{M}$ 
and $R_{T_x}$ be the upper $T_x$-induced relation by $\mathbf{M}$. 
Then $R^{P}=\bigcup_{x\in X}\{ x\}\times R^{P_x}$ is called the 
{\em lower $P$-induced state-transition relation} and 
$R_{T}=\bigcup_{x\in X}\{ x\}\times R_{T_x}$ is called the 
{\em upper $T$-induced state-transition relation}. 
The automaton ${\mathcal A}^{P}=(X,S,R^{P})$ is said to be 
the {\em lower $P$-induced automaton} and the automaton 
${\mathcal A}_{T}=(X,S,R_{T})$ is said to be 
the {\em upper $T$-induced automaton}.

We say that the automaton ${\mathcal A}$ {\em is recoverable from} $T_R$ ($P_R$) if, 
for all $x\in X$, $R_x$ {is recoverable from} $T_{R_x}$ ($P_{R_x}$), i.e., 
if ${\mathcal A}={\mathcal A}_{T_R}$ (${\mathcal A}={\mathcal A}^{P_R}$).

The following results follow immediately from  Lemma \ref{xreldreprest}, 
Theorem \ref{reldreprest} and Corollary \ref{fcorreldreprest}.

\begin{theorem}\label{labxreldreprest}
Let  $\mathbf{M}$ be a  non-trivial  complete lattice and  $S, X$ non-empty sets 
such that  $\mathbf{B}$ is a  bounded subposet of\/ $\mathbf{M}^{S}$. 
 Let $P:B\to ({M}^{S})^{X}$ and $T:B\to ({M}^{S})^{X}$ 
be morphisms of partially ordered   sets such that, for all $a, b\in B$ and all $x\in X$, 
$$
P_{x}(a)\leq b\ \ekviv\ a\leq T_{x}(b).
$$
\begin{enumerate}[{\rm(a)}]
\item If $P(B)\subseteq B^{X}$ then $R_T\subseteq R^{P}$.
\item If $T(B)\subseteq B^{X}$ then $R^{P}\subseteq R_T$.
\item If $P(B)\subseteq B^{X}$ and 
$T(B)\subseteq B^{X}$ then $R_T= R^{P}$ 
and ${\mathcal A}_{T}={\mathcal A}^{P}$.
\end{enumerate}
\end{theorem}

Hence, using Theorem \ref{labxreldreprest}, we can 
ask whether the functors computed by $(\star)$ and 
$(\star\star)$ can recover a given relation $R$ on the 
set of states. The answer is in the following theorem.

\begin{theorem}\label{relxxxdreprest}
Let  $\mathbf{M}$ be a  non-trivial  complete lattice and 
 $S, X$ non-empty sets equipped with a set of labelled 
transitions $R\subseteq X\times S\times S$. Let  $\mathbf{B}$ be  
a bounded subposet of\/ $\mathbf{M}^{S}$.  Let 
$P_R\colon B\to (M^{S})^{X}$ 
and $T_R:B\to (M^{S})^{X}$ 
be labelled transition functors  {constructed by means of} $R$.
Then, for all $a, b\in B$ and all $x\in X$, 
$$
P_{R_x}(a)\leq b\ \ekviv\ a\leq T_{R_x}(b).
$$
Moreover,  the following holds.
\begin{enumerate}[\rm(a)] 
 \item If $R=R_{T_R}$ and $T_R(B)\subseteq B^{X}$ then $R=R_{T_R}=R^{P_R}$. 
 \item If $R=R^{P_R}$ and $P_R(B)\subseteq B^{X}$ then $R=R_{T_R}=R^{P_R}$.
\end{enumerate}
\end{theorem}

The following corollary illustrates the situation in the case 
when our partially ordered set $\mathbf{B}$ of propositions is large enough, i.e., 
the case when $\{0,1\}^{S}\subseteq B$.

 \begin{corollary}\label{labfcorreldreprest}
Let  $\mathbf{M}$ be a  non-trivial  complete lattice and 
${\mathcal A}=(X,S,R)$ an automaton. Let 
$\mathbf{B}$ be a bounded subposet of\/ $\mathbf{M}^{S}$ 
such that $\{0,1\}^{S}\subseteq B$.  
Then the automaton ${\mathcal A}$ is recoverable both from $P_R$ and ${T_R}$. 
\end{corollary}

We can illustrate previous results in the following example.

\begin{example}\label{example2}\label{ex2} Consider the 
automaton ${\mathcal A}$, the set of propositions $B$  and the  state-transition 
relation  $R$ of Example \ref{firef}.  From Example \ref{expend2} 
we know the labelled upper transition functor 
$T_R=(T_{R_{x_1}}, T_{R_{x_2}})$ 
and the labelled lower transition functor 
$P_R=(P_{R_{x_1}}, P_{R_{x_2}})$ 
from $B$ to $(2^{S})^{X}$.
Since $B=2^{S}$ we have 
$T_{R_{x_1}}(B)\cup T_{R_{x_2}}(B) \subseteq B$ 
and $P_{R_{x_1}}(B)\cup P_{R_{x_2}}(B) \subseteq B$.

Now, we use $T_R$ for computing the 
transition relations $R_{T_{R_{x_1}}}$ and $R_{T_{R_{x_2}}}$
(by the formula $(\dagger)$ and Example \ref{expend4}) 
and  $P_R$ for computing the 
transition  relations $R^{P_{R_{x_1}}}$ and $R^{P_{R_{x_2}}}$ 
(by the formula $(\dagger\dagger)$  and Example \ref{expend4}). 
We obtain by Corollary \ref{fcorreldreprest} that $R_{T_{R_{x_1}}}=R^{P_{R_{x_1}}}=R_{x_1}$ and 
$R_{T_{R_{x_2}}}=R^{P_{R_{x_2}}}=R_{x_2}$. It follows that 
$R_{T_R}=R^{P_{R}}=\{x_1\}\times R_{T_{R_{x_1}}}\cup \{x_2\}\times R_{T_{R_{x_2}}}=R$,
 i.e., our given state-transition relation $R$ is simultaneously 
recoverable by the transition functors $T_R$ and $P_R$. 
Hence these functors are characteristics of the triple $(B,X,S)$. 
\end{example}

\section{Constructions of automata}

By a {\em synthesis} in Theory of Systems is usually meant 
the task to construct an automaton ${\mathcal A}$ which realizes 
a dynamic process at least partially known to the user. Hence, we are given 
a description of this dynamic process and we know the set  
$X$ of inputs. Our task is to set up the set $S$ of states 
and a relation $R$ on $S$  labelled by elements from 
$X$  such that the constructed automaton  $(X, S, R)$ 
induces the logic, i.e., the partially ordered set of propositions, 
which corresponds to the original description.

The algebraic tools collected in previous sections enable 
us to solve the mentioned task. In what follows we involve a 
construction of $S$ and $R$ provided our logic with the 
transition functor representing the dynamics of our system 
is given. As in the previous section, our logic ${\mathbf B}$ 
will be considered to be a bounded subposet $\mathbf B$ 
of a power ${\mathbf M}^S$  where ${\mathbf M}$ is 
a complete lattice of truth-values. Our logic ${\mathbf B}$ is equipped 
with a transition functor $T:B\to (M^{S})^{X}$ 
where $X$ is a set of possible inputs. We ask that 
either $T=T_{R}$ or $T=P^{R}$. 
Depending on the respective type 
of our considered logic and of the properties of $T$ we will present some partial solutions to this task.

\subsection{Automata via partially ordered sets}
\label{autopres}

Recall that (see e.g. \cite{Markowsky}), for any bounded   partially ordered  set $\mathbf{B}=({B};\leq, 0,1)$, we have 
 a full set $S_{\mathbf B}$ of morphisms  of bounded partially ordered set 
 into the two-element 
Boolean algebra considered as a bounded partially ordered set 
 ${\mathbf 2}=(\{0, 1\}; \leq, 0, 1))$. 
The elements $h_D: B\to \{0, 1\}$ of $S_{\mathbf B}$ (indexed by proper 
down-sets $D$ of  $\mathbf{B}$) are 
 morphisms of bounded partially ordered sets  defined by the 
prescription ${h_{D}}(a)=0$ iff $a\in D$. 

In other words, every bounded partially ordered  set  ${}{\mathbf B}$ can be 
embedded into a Boolean algebra ${\mathbf  2}^{S}$ for 
a certain set $S$ via the mapping $i_{{}{\mathbf B}}^{S}$.

Hence, it looks hopeful to use the bounded partially ordered  set 
${\mathbf 2}=(\{0, 1\}; \leq, 0, 1)$ 
for the construction of our state-transition relation 
$R_T\subseteq X\times S_{\mathbf B} \times S_{\mathbf B}$. 

As mentioned in the beginning of this section, 
we are interested in a construction of an automaton
${\mathcal A}=(X,S,R)$ for a given set $X$ of inputs and 
determined by a certain partially ordered set of propositions. 
We cannot assume that this set of propositions is necessarily 
a Boolean algebra. In the previous part we supposed that 
this logic ${\mathbf B}$ is a bounded partially ordered  set 
${\mathbf B}=(B,\leq,0,1)$. Now, we are going to solve the 
situation when it is only a subset $C$ of $B$.

\begin{theorem}\label{boolgaldreprest}
Let  $\mathbf{B}=({B};\leq, 0,1)$  be a bounded   partially ordered  set such that 
$\mathbf{B}$ is a bounded subposet of $2^{S_{\mathbf B}}$. Let 
$(C;\leq, 1)$ be a subposet of  $\mathbf{B}$ containing $1$, 
and $X$ a non-empty set.  
Let $T=(T_x)_{x\in X}$ where $T_{x}\colon{}C\to 2^{S_{\mathbf B}}$ 
are morphisms of   partially ordered sets such that  
$T_x(1)=1$ for all $x\in X$.  Let $R_T$ be the upper $T$-induced state-transition relation and 
$T_{R_T}\colon{}B\to (2^{S_{\mathbf B}})^{X}$ be the labelled upper transition functor 
constructed by means of the upper T-induced automaton 
${\mathcal A}_T=(X, S_{\mathbf B},R_T)$. Then, for all $b\in C$, 
$$T(b)=T_{R_T}(b).$$
\end{theorem}
\begin{proof} Clearly, $T_{R_T}=((T_{R_T})_{x})_{x\in X}$ 
where $(T_{R_T})_{x}:B\to 2^{S_{\mathbf B}}$ 
are morphisms of   partially ordered  sets  for all $x\in X$. 
We write $R_{T}=\bigcup_{x\in X}\{ x\}\times R_{T_x}$ where   
$R_{T_x}$, $x\in X$ are the upper 
$T_x$-induced relation by $\mathbf{2}$.

Let us choose $b\in C$ and $x\in X$ arbitrarily, but fixed. We have to check that 
$T_x(b)=(T_{R_T})_{x}(b)$. Assume that  
$s\in S_{\mathbf B}$. It is enough to verify that $T_x(b)(s)= \bigwedge\{b(t)\mid  s R_{T_x} t\}$.

Evidently, for all $t\in S_{\mathbf B}$ such that  
$s R_{T_x} t$, $T_x(b)(s) \leq b(t)$. 
Hence $T_x(b)(s)\leq \bigwedge\{b(t)\mid  s R_{T_x} t\}$. 
To get the other 
inequality assume that $T_x(b)(s)< \bigwedge\{b(t)\mid  s R_{T_x} t\}$. 
Then $T_x(b)(s)=0$ 
and $\bigwedge\{b(t)\mid  s R_{T_x} t\}=1$. 
Put $V_{x}=\{z\in B\mid (\exists y\in C)(T_x(y)(s)=1\ \text{and}\ y\leq z)\}$.  
It follows that 
$b\notin V_x$ and $V_x$ is an upper set of ${\mathbf B}$ such 
that $1\in V_x$ (since $T_x(1)(s)=1(s)=1$). 
Let $W_x$ be a maximal proper upper set of ${\mathbf B}$ 
including $V_x$ such that $b\notin W_x$. 
Put $U_x=B\setminus W_x$. Then $U_x$ is 
a proper down-set, $0\in U_x$, ${h_{U_x}}(b)=0$ 
and ${h_{U_x}}(z)=1$ for all 
$z\in V_x$, i.e., 
${h_{U_x}}\in S_{\mathbf B}$ such that 
$T_x(a)(s)\leq a({h_{U_x}})$ for all $a\in C$. 
But this yields that $s R_{T_x} h_{U_x}$, i.e., 
$1=\bigwedge\{b(t)\mid  s R_{T_x} t\}\leq b({h_{U_x}})={h_{U_x}}(b)=0$, 
a contradiction. \qed 
\end{proof}

Using the relation $R^P$ instead of $R_T$, 
we can obtain a statement dual to Theorem  \ref{boolgaldreprest}.

\subsection{Automata via Boolean algebras}
\label{autoboolpres}
As for bounded partially ordered sets we have that, for any   Boolean algebra 
${\mathbf B}=(B;\vee, \wedge, {}{'}, 0,$ $1)$, there is  
 a full set $S_{\mathbf B}^{\text{bool}}$ of morphisms of  Boolean algebras 
into the two-element 
Boolean algebra $\mathbf{2}=(\{0, 1\};\vee, \wedge, {}{'}, 0, 1)$.

In what follows, we will modify our Theorem \ref{boolgaldreprest} 
for the more special case when the considered 
subposet ${\mathbf C}$ is closed under finite infima.

We are now ready to show under which conditions our transition functor can be recovered.

\begin{theorem}\label{fullbooldreprest} 
Let  $\mathbf{B}=({B};\vee, \wedge, {}{'}, 0,1)$  be a Boolean algebra such 
that $\mathbf{B}$ is a sub-Boolean algebra of ${\mathbf  2}^{S_{\mathbf B}^{\text{bool}}}$. 
Let ${\mathbf C}=(C;\leq, 1)$ be a subposet of  $\mathbf{B}$ containing $1$ such that 
$x, y\in C$ implies $x\wedge y\in C$,
and $X$ a non-empty set.  
Let $T=(T_x)_{x\in X}$ where $T_{x}:C\to 2^{S_{\mathbf B}^{\text{bool}}}$ 
are mappings preserving finite meets such that  
$T_x(1)=1$ for all $x\in X$. Let $R_T$ be the upper $T$-induced state-transition relation and 
$T_{R_T}\colon{}B\to (2^{S_{\mathbf B}})^{X}$ be the labelled upper transition functor 
constructed by means of the upper T-induced automaton 
${\mathcal A}_T=(X, S_{\mathbf B}^{\text{bool}},R_T)$. Then, for all $b\in C$, 
$$T(b)=T_{R_T}(b).$$
\end{theorem}
\begin{proof} Let us choose $b\in C$ and $x\in X$ arbitrarily, but fixed. Assume that  
$s\in S_{\mathbf B}^{\text{bool}}$. 
As in Theorem \ref{boolgaldreprest} it is enough to verify that 
$T_x(b)(s)= \bigwedge\{b(t)\mid  s R_{T_x} t\}$. 

By the same considerations as in the proof of  Theorem \ref{boolgaldreprest}  
we have $T_x(b)(s)\leq \bigwedge\{b(t)\mid  s R_{T_x} t\}$. 
To get the other 
inequality assume that $T_x(b)(s)< \bigwedge\{b(t)\mid  s R_{T_x} t\}$. 
Then $T_x(b)(s)=0$ 
and $\bigwedge\{b(t)\mid  s R_{T_x} t\}=1$. 
Put $V_{x}=\{z\in B\mid (\exists y\in C)(T_x(y)(s)=1\ \text{and}\ y\leq z)\}$.  
It follows that 
$b\notin V_x$ and $V_x$ is a filter  of ${\mathbf B}$ such 
that $1\in V_x$ (since $y, z\in V_x\cap C$ implies 
$T_x(y\wedge z)(s)=(T_x(y)\wedge T_x(z))(s)=T_x(y)(s)\wedge T_x(z)(s)=1\wedge 1=1$  
and $T_x(1)(s)=1(s)=1$). 
Let $W_x$ be a maximal proper filter of ${\mathbf B}$ 
including $V_x$ such that $b\notin W_x$. 
Then $W_x$ is  an 
ultrafilter of  ${\mathbf B}$.  
The ultrafilter $W_x$ determines a map $g_{W_x}\in S_{{\mathbf B}}^{\text{bool}}$ 
such that ${g_{W_x}}(b)=0$ and ${g_{W_x}}(z)=1$ for all 
$z\in V_x$, i.e., ${g_{W_x}}\in S_{{\mathbf B}}^{\text{bool}}$ is 
such that $T_x(a)(s)\leq {g_{W_x}}(a)=a({g_{W_x}})$ for all $a\in C$. 
This yields that $s R_{T_x} g_{W_x}$, i.e., 
$1=\bigwedge\{b(t)\mid  s R_{T_x} t\}\leq b({g_{W_x}})={g_{W_x}}(b)=0$, a contradiction. \qed
\end{proof}

The example below shows an application of Theorem  \ref{fullbooldreprest}.  

\begin{example}\label{apthbool} Consider again the set $S=\{s_1, s_2, s_3\}$  of states, the 
set $X=\{x_1, x_2\}$, and the set of propositions $B=2^{S}$ 
of Example \ref{firef}. Recall that in this case $S=S_{\mathbf B}^{\text{bool}}$.

Assume that $C=\{0, r, p', q', 1\}\subseteq B$ from the logic  ${\mathbf B}$ of  
Example \ref{expend1}.

Assume further that our partially known transition operator $T$ from $C$ to $(2^{S})^{X}$ is given as follows: 

\medskip
\begin{center}
\begin{tabular}{c c}
\begin{tabular}{l l}
$T_{{x_1}}(0)=0$,&$T_{{x_1}}(1)=1$,\\
 $T_{{x_1}}(r)=r$,& $T_{{x_1}}(p')=q'$,\\
& $T_{{x_1}}(q')=p'$,\\
\end{tabular}&
\begin{tabular}{l l}
$T_{{x_2}}(0)=p$,&$T_{{x_2}}(1)=1$,\\
 $T_{{x_2}}(r)=1$,& $T_{{x_2}}(p')=1$,\\
& $T_{{x_2}}(q')=1$.\\
\end{tabular}\\
\end{tabular}
\end{center}

Note that $T$ was chosen as a restriction of the operator $T_R$ from Example \ref{expend2} on the set $C$. 

Then, by an easy computation, we obtain from 
($\dagger$) that $R_{T}=\{x_1\}\times R_{T_{x_1}}\cup \{x_2\}\times R_{T_{x_2}}$ where 
$$R_{T_{x_1}}=\{(s_1, s_2), (s_2,  s_1), (s_3,s_3)\}\ \text{and}\ 
R_{T_{x_2}}=\{(s_2, s_3),   (s_3, s_3)\}. 
$$

From Theorem \ref{fullbooldreprest} we have that $T$ is a restriction of the operator $T_{R_T}$ on the set $C$. 

Moreover, we can see that our state-transition relation $R$ from Example \ref{firef} coincides with the induced 
 state-transition relation $R_T$, i.e., our partially known transition operator $T$ gives us a full information 
 about the automaton ${\mathcal A}$ from Example \ref{firef}.
\end{example}

\section{Conclusion}

We have shown in our paper that to every automaton 
considered as an acceptor a certain dynamic logic can be assigned. 
The dynamic nature of an automaton is expressed via its transition 
relation labelled by inputs. The logic consists of propositions on the given 
automaton and its dynamic nature is expressed by means of the 
so-called transition functor. However, this logic enables us to derive 
again a certain relation on the set of states which is labelled by 
inputs. The main task is whether the relation derived from the logic 
and the transition functor is faithful, i.e., whether it coincides with 
the original transition relation of the automaton.

 In fact, we have shown that if our set of propositions is large enough 
 this recovering of the transition relation is possible. Several examples are included.

Conversely, having a set $B$ of propositions that describe behaviour 
of our intended automaton and the transition functor which 
express the dynamicity of this process together with the set $X$ of inputs 
(going from environment), we presented a construction of 
a set of states $S$ and of a state-transition relation $R$ on 
$S$ such the constructed automaton $(X,S,R)$ realizes the description given 
by the propositions. It is shown that for every large enough set 
of states the induced transition functor coincides 
with the original one.

We believe that this theory enables us to consider automata 
from a different point of view which is more close to 
logical treatment and which enables us to make estimations and 
forecasts of the behaviour of automaton particularly
in a nondeterministic mode. The next task will be to 
testify which type of automaton is determined by a suitable 
sort of logic.

\section*{Acknowledgement}

This is a pre-print of an article published in International Journal of Theoretical Physics. 
The final authenticated version of the article is available online at: 
\newline 
https://link.springer.com/article/10.1007/s10773-017-3311-0.

\end{document}